\documentclass[10pt]{article}
\usepackage{tikz}
\usepackage{amsmath}
\usepackage{amssymb}
\usepackage{color}
\usepackage{float}
\usepackage{graphicx}
\usepackage{hyperref}
\usepackage{multicol}
\usepackage{setspace}
\usepackage{subfigure}
\usepackage{verbatim}
\usepackage{authblk}

\newtheorem{theorem}{Theorem}
\newtheorem{lemma}{Lemma}

\newcommand{\sq}{\hbox{\rlap{$\sqcap$}$\sqcup$}}
\newcommand{\qed}{\hspace*{\fill}\sq}
\newenvironment{proof}{\noindent\textbf{Proof.}\ }{\qed\par\vskip 4mm\par}

\floatstyle{ruled}
\newfloat{algo}{thpb}{lop}
\floatname{algo}{Algorithm}

\title{Super-Fast Distributed Algorithms for Metric Facility Location}
\date{}

\author{Andrew Berns}
\author{James Hegeman} 
\author{Sriram V. Pemmaraju
\thanks{This work is supported in part by National Science Foundation grant
CCF 0915543. This is a full version of a paper that appeared in ICALP 2012 and includes proofs
missing from that paper due to space restrictions.}}

\affil{Department of Computer Science\\
The University of Iowa\\
Iowa City, Iowa 52242-1419, USA\\
\texttt{[andrew-berns,james-hegeman,sriram-pemmaraju]@uiowa.edu}
}

\begin{document}

\maketitle

\begin{abstract}
This paper presents a distributed $O(1)$-approximation algorithm, with
expected-$O(\log \log n)$ running time, in the $\mathcal{CONGEST}$ model for the
metric facility location problem on a size-$n$ clique network. Though metric
facility location has been considered by a number of researchers in low-diameter
settings, this is the first sub-logarithmic-round algorithm for the problem that
yields an $O(1)$-approximation in the setting of non-uniform facility opening
costs. In order to obtain this result, our paper makes three main technical
contributions. First, we show a new lower bound for metric facility location,
extending the lower bound of B\u{a}doiu et al. (ICALP 2005) that applies only to
the special case of uniform facility opening costs. Next, we demonstrate a
reduction of the distributed metric facility location problem to the problem of
computing an $O(1)$-ruling set of an appropriate spanning subgraph. Finally, we
present a sub-logarithmic-round (in expectation) algorithm for computing a
$2$-ruling set in a spanning subgraph of a clique. Our algorithm accomplishes
this by using a combination of randomized and deterministic sparsification.
\end{abstract}

\section{Introduction}
This paper explores the design of ``super-fast'' distributed algorithms in
settings in which bandwidth constraints impose severe restrictions on the volume
of information that can quickly reach an individual node. As a starting point
for our exploration, we consider networks of diameter one (i.e., cliques) so as
to focus on bandwidth constraints only and avoid latencies imposed by distance
between nodes in the network. We assume the standard $\mathcal{CONGEST}$ model
\cite{PelegBook}, which is a synchronous message-passing model in which each
node in a size-$n$ network can send a message of size $O(\log n)$ along each
incident communication link in each round. By ``super-fast'' algorithms we mean
algorithms whose running time is strictly sub-logarithmic, in any sense --
deterministic, in expectation, or with high probability (w.h.p.). Several
researchers have previously considered the design of such ``super-fast''
algorithms; see \cite{LenzenWattenhofer10,LPPP05,PattShamirTeplitsky11} for
recent examples of relevant results. The working hypothesis is that in
low-diameter settings, where congestion, rather than distance between nodes, is
the main bottleneck, we should be able to design algorithms that are much faster
than corresponding algorithms in high-diameter settings.

The focus of this paper is the \textit{distributed facility location} problem,
which has been considered by several researchers
\cite{GLS06,MoscibrodaWattenhofer05,PanditPemmaraju09,PanditPemmaraju10} in
low-diameter settings. We first describe the sequential version of the problem.
The input to the facility location problem consists of a set of
\textit{facilities} $\mathcal{F} = \{x_1, x_2, \ldots, x_m\}$, a set of
\textit{clients} $\mathcal{C} = \{y_1, y_2, \ldots, y_n\}$, an
\textit{opening cost} $f_i$ associated with each facility $x_i$, and a
\textit{connection cost} $D(x_i,y_j)$ between each facility $x_i$ and client
$y_j$. The goal is to find a subset $F \subseteq \mathcal{F}$ of facilities to
\textit{open} so as to minimize the facility opening costs plus connection
costs, i.e.,
\[FacLoc(F) := \sum_{x_i \in F} f_i + \sum_{y_j \in \mathcal{C}} D(F,y_j)\]
where $D(F,y_j) := \min_{x_i \in F} D(x_i,y_j)$. Facility location is an old and
well-studied problem in operations research
\cite{Balinski66,CNWBook,EHK77,HamburgerKuehn63,Stollsteimer63} that arises in
contexts such as locating hospitals in a city or locating distribution centers
in a region.

The \textit{metric facility location} problem is an important special case of
facility location in which the connection costs satisfy the following
``triangle inequality:'' for any $x_i, x_{i'} \in \mathcal{F}$ and
$y_j, y_{j'} \in \mathcal{C}$,
$D(x_i,y_j) + D(y_j,x_{i'}) + D(x_{i'},y_{j'}) \geq D(x_i,y_{j'})$. The facility
location problem, even in its metric version, is NP-complete and finding
approximation algorithms for the problem has been a fertile area of research. A
series of constant-factor approximation algorithms have been proposed for the
metric facility location problem, with a steady improvement in the approximation
factor. See \cite{Li11} for a recent 1.488-approximation algorithm. This result
is near-optimal because it is known \cite{GuhaKhuller98} that the metric
facility location problem has no polynomial-time algorithm yielding an
approximation guarantee better than 1.463 unless
$NP \subseteq DTIME(n^{O(\log \log n)})$. For non-metric facility location, a
simple greedy algorithm yields an $O(\log n)$-approximation, and this is also
optimal (to within a constant factor) because it is easy to show that the
problem is at least as hard as set cover.

More recently, the facility location problem has been used as an abstraction for
the problem of locating resources in a wireless network
\cite{FrankBook,PanditPemmarajuICDCN09}. Motivated by this application, several
researchers have considered the facility location problem in a distributed
setting. In \cite{MoscibrodaWattenhofer05,PanditPemmaraju09,PanditPemmaraju10},
the underlying communication network is a complete bipartite graph with
$\mathcal{F}$ and $\mathcal{C}$ forming the bipartition. At the beginning of the
algorithm, each node, whether it is a facility or a client, has knowledge of the
connection costs between itself and all nodes in the other part. In addition,
the facilities know their opening costs. In \cite{GLS06}, the underlying
communication network is a clique. Each node in the clique may choose to open as
a facility, and each node that does not open will connect to an open facility.
Note that all of the aforementioned work assumes the $\mathcal{CONGEST}$ model
of distributed computation. The facility location problem considered in
\cite{PanditPemmarajuICDCN09} assumes that the underlying communication network
is a \textit{unit disk graph} (UDG). The algorithm presented in that paper
ignores bandwidth constraints and works only in the $\mathcal{LOCAL}$ model
\cite{PelegBook}. While a UDG can have high diameter relative to the number of
nodes in the network, the authors \cite{PanditPemmarajuICDCN09} reduce the UDG
facility location problem to a collection of low-diameter facility location-type
problems, providing additional motivation for the current work.

None of the prior papers, however, achieve near-optimal approximation (i.e.,
constant-factor in the case of metric facility location and $O(\log n)$-factor
for non-metric facility location) in \textit{sub-logarithmic} rounds. While
\cite{GLS06} does present a \textit{constant-round}, constant-factor
approximation to metric facility location on a clique, it is only for the
special case of \textit{uniform} metric facility location, i.e., when all
facility opening costs are identical. The question that drives this paper, then,
is: Can we develop a distributed constant-factor approximation algorithm for the
metric facility location problem in the clique setting that runs in strictly
sub-logarithmic time? One can ask similar questions in the bipartite setting and
for non-metric facility location as well, but as a first step we focus on the
metric version of the facility location problem on a clique.

Distributed facility location is challenging even in low-diameter settings
because the input consists of $\Theta(n^2)$ information (there are $\Theta(n^2)$
connection costs), distributed across the network, which cannot quickly be
delivered to a single node (or even a small number of nodes) due to the
bandwidth constraints of the $\mathcal{CONGEST}$ model. Therefore, any fast
distributed algorithm for the problem must be truly distributed and must take
advantage of the available bandwidth, as well as structural properties of
approximate solutions. Also worth noting is that even though tight lower bounds
on the running times of distributed approximation algorithms have been
established \cite{KMW10}, none of these bounds extend to the low-diameter
setting considered in this paper. Thus, at the outset it was unclear if a
sub-logarithmic round algorithm providing a constant-factor approximation was
even possible for the facility location problem.

\paragraph{Main result.} The main result of this paper is an
$O(1)$-approximation algorithm, running in expected-$O(\log \log n)$ rounds in
the $\mathcal{CONGEST}$ model, for metric facility location on a size-$n$
clique. If the metric satisfies additional properties (e.g., it has constant
doubling dimension), then we obtain an $O(\log^* n)$-round $O(1)$-approximation
for the problem. Our results are achieved via a combination of techniques that
include (i) a new constant-factor lower bound on the optimal cost of metric
facility location and (ii) a randomized sparsification technique that leverages
the available bandwidth to (deterministically) process sparse subgraphs. For
ease of exposition, we assume that numbers in the input (e.g., connection and
opening costs) can each be represented in $O(\log n)$ bits and thus can be
communicated over a link in $O(1)$ rounds in the $\mathcal{CONGEST}$ model.

\subsection{Technical Overview of Contributions}

We start by precisely stating the distributed facility location problem on a
clique, as in \cite{MettuPlaxton03,GLS06}. Let $(X, D)$ be a discrete metric
space with point set $X = \{x_1, x_2, \ldots, x_n\}$. Let $f_i$ be the opening
cost of $x_i$. We view the metric space $(X, D)$ as a completely-connected
size-$n$ network $C = (X, E)$ with each point $x_i$ represented by a node (which
we also call $x_i$) and with $E$ representing the set of all pairwise
communication links. Each node $x_i$ knows $f_i$ and the connection costs
(distances) $D(x_i,x_j)$ for all $x_j \in X$. The problem is to design a
distributed algorithm that runs on $C$ in the $\mathcal{CONGEST}$ model and
produces a subset $F \subseteq X$ such that each node $x_i \in F$ opens and
provides services as a facility, and each node $x_i \notin F$ connects to the
nearest open node. The goal is to guarantee that
$FacLoc(F) \leq \alpha \cdot OPT$, where $OPT$ is the cost of an optimal
solution to the given instance of facility location and $\alpha$ is some
constant. We call this the \textsc{CliqueFacLoc} problem. Of course, we also
want our algorithm to be ``super-fast'' and terminate in $o(\log n)$ rounds. In
order to obtain the result described earlier, our paper makes three main
technical contributions.

\begin{enumerate}
\item\textbf{Reduction to an $O(1)$-ruling set problem.} Our first contribution
is an $O(1)$-round reduction of the distributed facility location problem on a
clique to the problem of computing an $O(1)$-ruling set of a specific spanning
subgraph of the clique $C$. Let $C' = (X, E')$ be a spanning subgraph of $C$. A
subset $Y \subseteq X$ is said to be \textit{independent} if no two nodes in $Y$
are neighbors in $C'$. An independent set $Y$ is a
\textit{maximal independent set} (MIS) if no superset $Y' \supset Y$ is
independent in $C'$. An independent set $Y$ is $\beta$-ruling if every node in
$X$ is at most $\beta$ hops along edges in $C'$ from some node in $Y$. Clearly,
an MIS is a $1$-ruling set. We describe an algorithm that approximates
distributed facility location on a clique by first computing a spanning subgraph
$C'$ in $O(1)$ rounds. Then we show that a solution to the \textsc{CliqueFacLoc}
problem (i.e., a set of nodes to open) can be obtained by computing a
$\beta$-ruling set in $C'$ and then selecting a certain subset of the ruling
set. This step -- selecting an appropriate subset of the $\beta$-ruling set --
can also be accomplished in $O(1)$ rounds. The parameter $\beta$ affects the
approximation factor of the computed solution and we show that enforcing
$\beta = O(1)$ ensures that the solution to facility location is an
$O(1)$-approximation.

\begin{figure}
\begin{center}
\begin{tikzpicture}[scale=1,auto,swap]
    \draw (0,0) circle (3cm);
    \draw (0,0) coordinate (i);
    \draw (60:25mm) coordinate (v);
    \draw (135:1cm) coordinate (w);
    \draw (180:15mm) coordinate (x);
    \draw (270:2cm) coordinate (y);
    \draw (300:3cm) coordinate (z);
    \path[draw,thick,-] (i) -- node[below] {$r_i$} (0:3cm);
    \draw (v) -- (60:3cm);
    \draw (w) -- (135:3cm);
    \draw (x) -- (180:3cm);
    \draw (y) -- (270:3cm);
    \node[circle,draw,fill=black!100,inner sep=0pt,minimum width=3pt] at (i) {};
    \node[below] at (i) {$x_i$};
    \node[circle,draw,fill=black!100,inner sep=0pt,minimum width=3pt] at (v) {};
    \node[below left] at (v) {$x_j$};
    \node[circle,draw,fill=black!100,inner sep=0pt,minimum width=3pt] at (w) {};
    \node[right] at (w) {$x_k$};
    \node[circle,draw,fill=black!100,inner sep=0pt,minimum width=3pt] at (x) {};
    \node[below] at (x) {$x_{\ell}$};
    \node[circle,draw,fill=black!100,inner sep=0pt,minimum width=3pt] at (y) {};
    \node[left] at (y) {$x_p$};
    \node[circle,draw,fill=black!100,inner sep=0pt,minimum width=3pt] at (z) {};
    \node[below right] at (z) {$x_q$};
\end{tikzpicture}
\end{center}
\caption{This is an illustration of a radius-$r_i$ ball centered at $x_i$. There
are 6 points (including $x_i$) inside this ball, implying that the sum of 6 
``distances,'' denoted by line segments from points to the ball-boundary, equals
$f_i$.}
\label{fig:ri}
\end{figure}
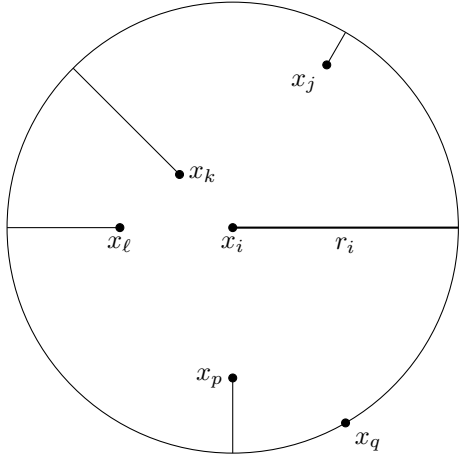

\item\textbf{A new lower bound for metric facility location.} To show that the
computation of an $O(1)$-ruling set, as sketched above, does indeed lead to an
$O(1)$-approximation algorithm for \textsc{CliqueFacLoc}, we develop new
analysis tools. In particular, we derive a new lower bound on the cost of an
optimal solution to the facility location problem. For $x \in X$, let $B(x,r)$
denote the set of points $y \in X$ satisfying $D(x,y) \leq r$. For each $x_i$,
let $r_i$ be the nonnegative real number satisfying
\[\sum\limits_{y \in B(x_i,r_i)} (r_i - D(x_i,y)) = f_i.\]
See Figure \ref{fig:ri} for intuition regarding this definition of the $r_i$'s.
As observed by Mettu and Plaxton \cite{MettuPlaxton03}, $r_i$ exists and is
uniquely defined. B\u{a}doiu et al. proved in \cite{BCIS05} that
$\sum_{i=1}^n r_i$ is a constant-factor approximation for $OPT$ in the case of
\textit{uniform} facility opening costs; this fact plays a critical role in the
design of the constant-round, constant-factor approximation algorithm of
Gehweiler et al. \cite{GLS06} for the special case of \textsc{CliqueFacLoc} in
which all facility opening costs are identical. However, the sum
$\sum_{i=1}^n r_i$ can be arbitrarily large in relation to $OPT$ when the
$f_i$'s are allowed to vary.
\begin{figure}
\begin{center}
\begin{tikzpicture}[scale=1,auto,swap]
    \draw (0,0) coordinate (a);
    \draw (3,0) coordinate (b);
    \path[draw,-] (a) -- node[above] {$1$} (b);
    \node[circle,draw,fill=black!100,inner sep=0pt,minimum width=3pt] at (a) {};
    \node[below] at (a) {$f_1 = 1$};
    \node[circle,draw,fill=black!100,inner sep=0pt,minimum width=3pt] at (b) {};
    \node[below] at (b) {$f_2 = 99$};
\end{tikzpicture}
\end{center}
\caption{Here $r_1 = 1$ and $r_2 = 50$. However, the optimal solution involves
opening only point $x_1$ and costs only $2$ units. The sum $r_1 + r_2$ can be
made arbitrarily large relative to the optimal cost by simply increasing $f_2$.
Note also that $\overline{r}_2$ is just $2$.}
\label{fig:example}
\end{figure}
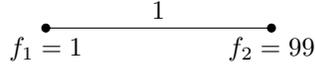
Consider an example consisting of only two nodes, one of whose opening costs is
large in comparison to the other and to the distance between them. (See Figure
\ref{fig:example}.) Though the $r_i$'s turn out not to \textit{directly} provide
a lower bound, they are still quite useful. We apply the following (idempotent)
transformation
\[r_i \rightarrow \overline{r}_i = \min_{1 \leq j \leq n} \{D(x_i,x_j) + r_j\}\]
to define, for each $x_i$, a new quantity that we call $\overline{r}_i$, and use
$\overline{r}_i$ instead of $r_i$ to formulate a lower bound. Note that for any
$i$, $\overline{r}_i \leq r_i$. In the example in Figure \ref{fig:example},
$r_2 = 50$, but $\overline{r}_2 = 2$. We show later that
$\sum_{i=1}^n \overline{r}_i$ bounds the optimal cost $OPT$ from below (to
within a constant factor) in the general case of non-uniform facility opening
costs (Lemma \ref{lemma:lower_bound}). We complete our analysis by showing that
using an $O(1)$-ruling set produces a solution to \textsc{CliqueFacLoc} whose
cost is bounded above by a constant times $\sum_{i=1}^n \overline{r}_i$ (Lemma
\ref{lemma:upper_bound}).

\item\textbf{An $O(1)$-ruling set via a combination of randomized and
deterministic sparsification.} Our final contribution is an
expected-$O(\log \log n)$-round algorithm for computing a $2$-ruling set of a
given spanning subgraph $C'$ of a clique $C$. We start by describing a
deterministic ``subroutine'' that takes a subset $Z \subseteq X$ as input and
computes an MIS of $C'[Z]$ (i.e., the subgraph of $C'$ induced by $Z$) in $c$
rounds if $C'[Z]$ has at most $c \cdot n$ edges. This is achieved via a simple
load-balancing scheme that communicates the entire subgraph $C'[Z]$ to all nodes
in $c$ rounds. We then show how to use randomization to repeatedly peel off
subgraphs with linearly many edges (in expectation) for processing by the
aforementioned subroutine. In this manner, the entire graph $C'$ can be
processed using a number of subroutine calls which is $O(\log \log n)$ in
expectation (Theorem \ref{theorem:2rule_time}).
\end{enumerate}

\subsection{Related Work}
\label{section:related}

In \cite{MoscibrodaWattenhofer05}, Moscibroda and Wattenhofer use the technique
of distributed LP-rounding to solve the facility location problem in the
$\mathcal{CONGEST}$ model, assuming that the communication network is the
complete bipartite graph $G = (\mathcal{F}, \mathcal{C}, E)$. Let
$m = |\mathcal{F}|$ and $n = |\mathcal{C}|$. Assuming that the connection costs
and facility opening costs have size that is polynomial in $(m + n)$, they
achieve, for every constant $k$, an
$O(\sqrt{k} (m n)^{1 / \sqrt{k}} \log (m + n))$-approximation in $O(k)$
communication rounds. Note that one can obtain ``super-fast'' algorithms for
facility location by taking $k$ small enough, but with a significant
corresponding loss in the approximation factor. For the metric facility location
problem, Pemmaraju and Pandit use the primal-dual method to obtain a
$7$-approximation \cite{PanditPemmaraju09} that runs in $O(\log m + \log n)$
rounds. The same paper contains a generalization of this result: A $k$-round
algorithm that, for every constant $k$, yields an approximation factor of
$O(m^{2 / \sqrt{k}} \cdot n^{3 / \sqrt{k}})$. Note that setting
$k = \log^2 (m n)$ here yields an $O(1)$-approximation in $O(\log^2 (m n))$
rounds for the metric facility location problem. Subsequently
\cite{PanditPemmaraju10}, Pemmaraju and Pandit improved the
Moscibroda-Wattenhofer result for non-metric facility location. For instances
with $m$ facilities and $n$ clients, for any positive $k$, the algorithm in
\cite{PanditPemmaraju10} runs in $O(k)$ rounds and yields a
$O((m n)^{5 / \sqrt{k}} \cdot \log n)$-approximation, shaving off a
``logarithmic'' term from the approximation factor achieved by Moscibroda and
Wattenhofer.

While all of the above mentioned distributed algorithms are fast and achieve
near-optimal approximation ratios for facility location, none of them seem to
take particular advantage of the small diameter of the network on which they are
executing. An excellent illustration of a distributed algorithm that takes
advantage of the low-diameter setting in which it operates is provided by the
minimum spanning tree (MST) algorithm of Lotker et al.~\cite{LPPP05}. Consider a
clique network in which each edge $(u,v)$ has an associated weight $w(u,v)$ of
which only nodes $u$ and $v$ are aware. The problem is for the nodes to compute
an MST of the edge-weighted clique such that after the computation, each node
knows the MST edges incident on it. It is important to note that the problem is
defined by $\Theta(n^2)$ inputs and it would take
$\Omega\left(\frac{n}{\log n}\right)$ rounds of communication for all of this
information to reach a single node (in the $\mathcal{CONGEST}$ model of
distributed computation). Lotker et al.~\cite{LPPP05} showed that the MST
problem on a clique can in fact be solved in $O(\log \log n)$ rounds in the
$\mathcal{CONGEST}$ model. The algorithm of Lotker et al.~employs a clever
merging procedure that causes the sizes of the MST components to, roughly
speaking, square with each iteration, leading to an $O(\log \log n)$-round
computation time. 

Several more recent papers have continued the development of ``super-fast''
algorithms in low-diameter settings. In STOC 2011, Lenzen and Wattenhofer
\cite{LenzenWattenhofer11} derived tight bounds on parallel load balancing and
their result has applications in how information can be quickly disseminated in
a clique (in the $\mathcal{CONGEST}$ model). In PODC 2011, Patt-Shamir and
Teplitsky \cite{PattShamirTeplitsky11} presented on $O(\log \log n)$ randomized
algorithm for the \textit{distributed sorting} problem. Third, Lenzen
\cite{Lenzen12} showed that randomization is not necessary for solving problems
such as distributed sorting efficiently. Lenzen presented
\textit{deterministic}, constant-round algorithms for a routing problem and for
the distributed sorting problem considered in \cite{PattShamirTeplitsky11}.
Constant-round algorithms for sophisticated problems, of the kind described by
Lenzen \cite{Lenzen12}, highlight the difficulty of showing non-trivial lower
bounds in the $\mathcal{CONGEST}$ model for clique networks. For example, it has been proved that
computing an MST in general requires $\Omega(\sqrt[4]{n} / \log n)$ rounds for
diameter-$3$ graphs \cite{LPP06}, but no non-trivial lower bounds are known for
diameter-$2$ or clique (diameter-$1$) networks.

\section{Reduction to the $O(1)$-Ruling Set Problem}
\label{sect:FacLoc}

\subsection{Algorithm}

We present our distributed algorithm for \textsc{CliqueFacLoc} in Algorithm
\ref{alg:FacLocAlg}. This algorithm is not complete in the sense that it does
not solve \textsc{CliqueFacLoc} directly, but rather reduces it to a problem of
computing an $s$-ruling set on a spanning subgraph of the clique network. We
complete the algorithm in the next section by presenting a $2$-ruling set
algorithm that runs in expected-$O(\log \log n)$ rounds.

\begin{algo}
\textbf{Input:} A discrete metric space of nodes $(X, D)$, with opening costs;\\
a sparsity parameter $s$\\
\textbf{Assumption:} Each node knows its own opening cost and the distances from
itself to other nodes\\
\textbf{Output:} A subset of nodes (a \textit{configuration}) to be declared
open
\begin{tabbing}
......\=a....\=b....\=c....\=d....\=e....\=f....\=g....\=h......\kill
1.\>Each node $x_i$ computes and broadcasts its value $r_i$;
$r_0 := \min_i r_i$.\\
2.\>Each node computes a partition of the network into classes ${V_k}$,
$k=0, 1, \ldots$ with\\
\>\>$c_0^k \cdot r_0 \leq r_j < c_0^{k+1} \cdot r_0$ for $x_j \in V_k$.\\
3.\>Each node $x_i \in V_k$ determines its neighbors within its own class $V_k$
using the following rule:\\
\>\>For $x_j \in V_k$, $x_j$ is a neighbor of $x_i$ if and only if
$D(x_i,x_j) \leq r_i + r_j$.\\
\>\>The graph on vertex set $V_k$ induced by these edges is denoted $H_k$.\\
4.\>All nodes now use procedure \textsc{RulingSet}($\bigcup_k H_k$,$s$) to
determine\\
\>\>an $s$-ruling set $T^* \subseteq X$. We use $T_k$ to denote
$T^* \cap V_k$.\\
5.\>Each node $x_i$ broadcasts its membership status with respect to the
$s$-ruling set\\
\>\>of its class, $T_k$.\\
6.\>A node $x_i \in V_k$ declares itself to be open if:\\
\>\>(i) $x_i$ is a member of set $T_k \subseteq V_k$, and\\
\>\>(ii) There is no node $x_j$ belonging to a class $V_{k'}$, with $k' < k$,\\
\>\>\>such that $D(x_i,x_j) \leq 2 r_i$.\\
7.\>Each node broadcasts its status (open or not), and nodes connect to the
nearest open facility.
\end{tabbing}
\caption{\textsc{FacilityLocation}}
\label{alg:FacLocAlg}
\end{algo}

\noindent
Algorithm \ref{alg:FacLocAlg} consists of three stages, which we now describe.

\noindent\textit{Stage 1 (Steps 1-2).} Each node knows its own opening cost and
the distances to other nodes, so node $x_i$ computes $r_i$ and broadcasts that
value to all others. Once this is complete, each node knows all of the $r_i$
values. Next, every node computes a partition of the network into groups whose
$r_i$ values vary by at most a factor of $c_0 = 1 + \frac{1}{\sqrt{2}}$
(Step 2). Specifically, let $r_0 := \min_{1 \leq j \leq n} \{r_j\}$, and define
the class $V_k$ to be the set of nodes $x_i$ such that
$c_0^k \cdot r_0 \leq r_i < c_0^{k+1} \cdot r_0$. Every node computes the class
into which each node in the network, including itself, falls.
\vspace{2mm}

\noindent\textit{Stage 2 (Steps 3-5).} We now focus our attention on class
$V_k$. Suppose $x_i, x_j \in V_k$. We define $x_i$ and $x_j$ to be
\textit{adjacent} in class $V_k$ if $D(x_i,x_j) \leq r_i + r_j$. Each node in
$V_k$ can determine its neighbors in $V_k$. We refer to the graph on nodes in
$V_k$ induced by this adjacency condition as $H_k$. Next, consider the spanning
subgraph (on all $n$ nodes) $\bigcup_k H_k$. We apply procedure
\textsc{RulingSet}() to $\bigcup_k H_k$ to compute an $s$-ruling set $T^*$ of
$\bigcup_k H_k$. We describe a ``super-fast'' (in expectation) implementation of
\textsc{RulingSet}($\cdot$,$2$) in Section \ref{sect:2Ruling}. An $s$-ruling set
$T^* \subseteq X$ of $\bigcup_k H_k$ determines, for each $k$, an $s$-ruling set
$T_k \subseteq V_k$ of $H_k$. After the sparse sets $T_k$ have been constructed
for the classes $V_k$, each node broadcasts its membership status with respect
to the $s$-ruling set $T_k$ of its own class.
\vspace{2mm}

\noindent\textit{Stage 3 (Steps 6-7).} Finally, a node $x_i$ in class $V_k$
opens if (i) $x_i \in T_k$, and (ii) there is no node $x_j \in B(x_i, 2 r_i)$ of
a class $V_{k'}$ with $k' < k$. Open facilities declare themselves via
broadcast, and every node connects to the nearest open facility.

\subsection{Running Time Analysis}

The accounting of the number of communication rounds required by Algorithm
\ref{alg:FacLocAlg} is straightforward. Stage 1 requires exactly one round of
communication, to broadcast $r_i$ values. Stage 2 requires $O(\mathcal{T}(n,s))$
rounds to compute the $s$-ruling subsets $\{T_k\}_k$, and an additional round to
broadcast membership status. Stage 3 requires one round, in order to inform
others of a nodes decision to open or not. Thus, the running time of our
algorithm in communication rounds is $O(\mathcal{T}(n,s))$. In Section
\ref{sect:2Ruling} we show that $\mathcal{T}(n,2)$ can be $O(\log \log n)$ in
expectation.

%Lemma 1
\begin{lemma}
Algorithm \ref{alg:FacLocAlg} runs in $O(\mathcal{T}(n,s))$ rounds, where
$\mathcal{T}(n,s)$ is the number of communication rounds needed to compute an
$s$-ruling set of a spanning subgraph $C'$ of the $n$-node clique network.
\label{lemma:facloc_running_time}
\end{lemma}

\subsection{Cost Approximation Analysis}

We now show that Algorithm \ref{alg:FacLocAlg} produces an $O(s)$-approximation
to \textsc{CliqueFacLoc}. This analysis borrows ideas from the analysis of a
simple, greedy, sequential facility location algorithm due to Mettu and Plaxton
\cite{MettuPlaxton03}. The Mettu-Plaxton algorithm considers points $x_i$ in
non-decreasing order of the $r_i$'s. Then, each $x_i$ under consideration is
included in the solution if $B(x_i, 2 r_i)$ does not contain any point already
included in the solution. Mettu and Plaxton show that, if $F_{MP} \subseteq X$
is the set of facilities opened by their algorithm,
$FacLoc(F_{MP}) \leq 3 \cdot OPT$.

We next recall the charging scheme employed by Mettu and Plaxton for the
analysis of their algorithm. The $charge(\cdot,\cdot)$ of a node $x_i$ with
respect to a collection of (open) facilities $F$ (also known as a
\textit{configuration}) is defined by
\[charge(x_i,F) = D(x_i,F) + \sum\limits_{x_j \in F} \max \{0, r_j - D(x_j,x_i)\}\]
where $D(x_i,F) = \min_{x_j \in F} D(x_i,x_j)$. It is easy to check that the
cost of a configuration $F$, $FacLoc(F)$, is precisely equal to the sum of the
charges with respect to $F$, i.e., $\sum_{i=1}^n charge(x_i,F)$
\cite{MettuPlaxton03}. Given that the Mettu-Plaxton algorithm yields a
$3$-approximation, we see that for any $F \subseteq X$,
\[FacLoc(F) \geq \frac{1}{3} FacLoc(F_{MP}) = \frac{1}{3} \sum_{i=1}^n charge(x_i,F_{MP})\]

The rest of our analysis consists of two parts. In the first part, we show (as
promised) that $\sum_{i=1}^n \overline{r}_i$ is a constant-factor lower bound
for $OPT$. In the second part, we show the corresponding upper bound result. In
other words, we show that for the subset $F^*$ of facilities opened by Algorithm
\ref{alg:FacLocAlg}, $FacLoc(F^*) = O(\sum_{i=1}^n \overline{r}_i)$.

\subsubsection{A New Lower Bound for Non-uniform Metric Facility Location.}

%Lemma 2
\begin{lemma}
$FacLoc(F) \geq (\sum_{i=1}^n \overline{r}_i) / 6$ for any configuration $F$.
\label{lemma:lower_bound}
\end{lemma}
\begin{proof}
Notice that $F_{MP}$ has the property that no two facilities
$x_i, x_j \in F_{MP}$ can be so close that $D(x_i,x_j) \leq r_i + r_j$
\cite{MettuPlaxton03}. Therefore, if $x_{\delta(i)}$ denotes a closest open
facility (i.e., an open facility satisfying
$D(x_i,x_{\delta(i)}) = D(x_i,F_{MP})$), then
\begin{align*}
FacLoc(F_{MP}) &= \sum\limits_{i=1}^n charge(x_i,F_{MP})\\[1mm]
&= \sum\limits_{x_j \in F_{MP}} charge(x_j,F_{MP})
+ \sum\limits_{x_i \notin F_{MP}} charge(x_i,F_{MP})\\[1mm]
&\geq \sum\limits_{x_j \in F_{MP}} r_j + %
\sum\limits_{x_i \notin F_{MP}} \left[D(x_i,x_{\delta(i)}) + %
\max \{0, r_{\delta(i)} - D(x_{\delta(i)},x_i)\}\right]\\[1mm]
&= \sum\limits_{x_j \in F_{MP}} r_j
+ \sum\limits_{x_i \notin F_{MP}} \max \{r_{\delta(i)}, D(x_i,x_{\delta(i)})\}\\
\end{align*}
Note that the inequality in the above calculation (in the third line) follows
from observing that $charge(x_j,F_{MP}) \geq r_j$ for $x_j \in F_{MP}$, and from
throwing away some terms of the sum in the definition of $charge(x_i,F_{MP})$
for $x_i \notin F_{MP}$.

Now, recall the definition
$\overline{r}_i = \min_{1 \leq j \leq n} \{D(x_i,x_j) + r_j\}$. Therefore,
$\overline{r}_i \leq r_i$, and $\overline{r}_i \leq D(x_i,x_{\delta(i)}) +
r_{\delta(i)} \leq 2 \cdot \max \{r_{\delta(i)}, D(x_i,x_{\delta(i)})\}$. It
follows that
\begin{align*}
FacLoc(F_{MP}) &\geq \sum\limits_{x_j \in F_{MP}} \overline{r}_j
+ \sum\limits_{x_i \notin F_{MP}} \frac{\overline{r}_i}{2}\\[1mm]
&\geq \sum\limits_{x_j \in F_{MP}} \frac{\overline{r}_j}{2}
+ \sum\limits_{x_i \notin F_{MP}} \frac{\overline{r}_i}{2}\\[1mm]
&= \frac{1}{2} \cdot \sum\limits_{i=1}^n \overline{r}_i\\
\end{align*}
Therefore
$FacLoc(F) \geq FacLoc(F_{MP}) / 3 \geq (\sum_{i=1}^n \overline{r}_i) / 6$, for
any configuration $F$.
\end{proof}

\subsubsection{The Upper Bound Analysis}

Let $F^*$ be the set of nodes opened by our algorithm. We analyze $FacLoc(F^*)$
by bounding $charge(x_i,F^*)$ for each $x_i$. Recall that
$FacLoc(F) = \sum_{i=1}^n charge(x_i,F)$ for any $F$. Since $charge(x_i,F^*)$ is
the sum of two terms, $D(x_i,F^*)$ and
$\sum_{x_j \in F^*} \max \{0, r_j - D(x_j,x_i)\}$, bounding each term separately by a
$O(s)$-multiple of $\overline{r}_i$, yields the result.

In the following analysis, we mainly use the property of an
$s$-ruling set $T_k \subseteq V_k$ that for any node
$x_i \in V_k$, $D(x_i,T_k) \leq 2 c_0 r_i \cdot s$.
%Also, for no two
%members of $T_k$ is the distance between them less than $2 c_0^{k} r_0$. 
Note that here we are using distances from the metric $D$ of $(X, D)$.
%Now, in our cost analysis, we consider a node $x_i \in V_k$. To bound
%$D(x_i,F^*)$, observe that either $x_i \in T_k$, or else there exists a node
%$x_j \in T_k$ such that
%$D(x_i,x_j) \leq 2 c_0 r_i \cdot s \leq (2 + \sqrt{2}) r_i \cdot s$. 
We also make critical use of the property of our algorithm that if a node $x_j \in T_k$ does not open, then there exists another node $x_{j'}$ in a
class $V_{k'}$, with $k' < k$, such that $D(x_j,x_{j'}) \leq 2 r_j$.

%We are now ready to bound the components of $charge(x_i,F^*)$.

%Lemma 3
\begin{lemma}
$D(x_i,F^*) \leq (s + 1) \cdot 4 c_0^2 \cdot \overline{r}_i$.
\label{lemma:open_fac_dist}
\end{lemma}
\begin{proof}
Let $x_{i'}$ be a minimizer for $D(x_i,x_y) + r_y$ (where $x_{i'}$ may be $x_i$
itself), so that $\overline{r}_i = D(x_i,x_{i'}) + r_{i'}$. Suppose that
$x_{i'} \in V_{k'}$. Note that $k' \le k$.
We know that $x_{i'}$ is within distance
$2 c_0 s \cdot r_{i'}$ of a node $x_{j'} \in T_{k'}$ (which may be $x_{i'}$
itself). Then, $x_{j'}$ either opens, or there exists a node $x_{j_1}$ of a
lower class such that $D(x_{j'},x_{j_1}) \leq 2 r_{j'}$. In the former case,
$D(x_{i'},F^*) \leq 2 c_0 s \cdot r_{i'}$; in the latter case we have
$D(x_{i'},x_{j_1}) \leq D(x_{i'},x_{j'}) + D(x_{j'},x_{j_1}) \leq
2 c_0 s \cdot r_{i'} + 2 r_{j'} \leq (s + 1) \cdot 2 c_0 r_{i'}$, the last
inequality owing to the fact that $x_{i'}$ and $x_{j'}$ belong to the same
class.

So, within a distance $(s + 1) \cdot 2 c_0 r_{i'}$ of $x_{i'}$, there exists
either an open node or a node of a lower class. In the latter case (in which
there is a node $x_{j_1}$ of a lower class), we repeat the preceding analysis
for $x_{j_1}$; within a distance $(s + 1) \cdot 2 c_0 r_{j_1}$ of $x_{j_1}$,
there must exist either an open node or a node of a class $V_{k_2}$, where
$k_2 \leq k' - 2$.

Repeating this analysis up to $k' + 1$ times shows that, within a distance of at
most $(s + 1) \cdot 2 c_0 \cdot (r_{i'} + r_{j_1} + r_{j_2} + r_{j_3} + \ldots +
r_{j_{k'}})$, where $r_{j_w}$ is the characteristic radius of a node $x_{j_w}$
in class $V_{k'-w}$, there exists a node which opens as a facility. This
distance is naturally bounded above by $(s + 1) \cdot 2 c_0 \cdot
(r_{i'} + r_{i'} + \frac{1}{c_0} r_{i'} + \frac{1}{c_0^2} r_{i'} + \ldots)
\leq (s + 1) \cdot 2 c_0 \cdot (2 + \sqrt{2}) r_{i'} =
(s + 1) \cdot 4 c_0^2 \cdot r_{i'}$.
Therefore,
\begin{align*}
D(x_i,F^*) &\leq D(x_i,x_{i'}) + D(x_{i'},F^*)\\[1mm]
&\leq D(x_i,x_{i'}) + (s + 1) \cdot 4 c_0^2 \cdot r_{i'}\\[1mm]
&\leq (s + 1) \cdot 4 c_0^2 \cdot \left(D(x_i,x_{i'}) + r_{i'}\right)\\[1mm]
&= (s + 1) \cdot 4 c_0^2 \cdot \overline{r}_i\\
\end{align*}
\end{proof}

%Lemma 4
\begin{lemma}
$\sum_{x_j \in F^*} \max \{0, r_j - D(x_j,x_i)\} \leq c_0 \cdot \overline{r}_i$.
\label{lemma:open_fac_contrib}
\end{lemma}
\begin{proof}
We begin by observing that we cannot simultaneously have $D(x_j,x_i) \leq r_j$
and $D(x_l,x_i) \leq r_l$ for $x_j, x_l \in F^*$ and $j \neq l$. Indeed, if this
were the case, then $D(x_j,x_l) \leq r_j + r_l$. If $x_j$ and $x_l$ were in the
same class $V_y$, then they would be adjacent in $H_y$; this is impossible, for
then they could not both be members of $T_y$ (for a node in $V_y$, membership in
$T_y$ is necessary to join $F^*$). If $x_j$ and $x_l$ were in different classes,
assume WLOG that $r_j < r_l$. Then $D(x_j,x_l) \leq r_j + r_l \leq 2 r_l$, and
$x_l$ should not have opened. These contradictions imply that there is at most
one node $x_j \in F^*$ for which $D(x_j,x_i) \leq r_j$.

For the rest of this lemma, then, assume that $x_j \in F^*$ is the unique open
node such that $D(x_j,x_i) \leq r_j$ (if such a $x_j$ does not exist, there is
nothing to prove). Note that $x_i$ cannot be of a lower class than $x_j$ (for
else $x_j$ would not have opened). Consequently, $r_j < c_0 \cdot r_i$.

Now, suppose that $c_0 \overline{r}_i < r_j - D(x_j,x_i)$. As before, let
$x_{i'}$ be a minimizer for $D(x_i,x_y) + r_y$ (where $x_{i'}$ may be $x_i$
itself). Then $c_0 \cdot D(x_i,x_{i'}) + c_0 \cdot r_{i'} < r_j - D(x_j,x_i)$,
so we have (i) $c_0 r_{i'} < r_j$ (and $x_{i'}$ is in a lower class than $x_j$)
and (ii) $D(x_j,x_{i'}) \leq D(x_j,x_i) + D(x_i,x_{i'}) < r_j$, which means
that $x_j$ should not have opened. This contradiction completes the proof of the
lemma.
\end{proof}

%Lemma 5
\begin{lemma}
$FacLoc(F^*) \leq 6 \cdot (4 c_0^2 s + 4 c_0^2 + c_0) \cdot OPT.$
\label{lemma:upper_bound}
\end{lemma}
\begin{proof}
Combining the Lemmas \ref{lemma:open_fac_dist} and \ref{lemma:open_fac_contrib}
gives
\begin{align*}
FacLoc(F^*) &= \sum\limits_{i=1}^n charge(x_i,F^*) = \sum\limits_{i=1}^n
\left[D(x_i,F^*) + \sum_{x_j \in F^*} \max \{0, r_j - D(x_j,x_i)\}\right]\\[1mm]
&\leq \sum\limits_{i=1}^n
\left[(4 c_0^2 s + 4 c_0^2) \cdot \overline{r}_i + c_0 \overline{r}_i\right]
\leq (4 c_0^2 s + 4 c_0^2 + c_0) \cdot \sum\limits_{i=1}^n \overline{r}_i\\
\end{align*}

\noindent
From Lemma \ref{lemma:lower_bound}, we know that $\sum_{i=1}^n \overline{r}_i \le 6 \cdot OPT$.
Combining this fact with the above inequality, yields the result.
\end{proof}

\noindent
Lemma \ref{lemma:facloc_running_time} on the running time of the algorithm, combined with the above lemma on the approximation 
factor, yield the following result.
%Theorem 1
\begin{theorem}
Algorithm \ref{alg:FacLocAlg} (\textsc{FacilityLocation}) computes an
$O(s)$-factor approximation to \textsc{CliqueFacLoc} in $O(\mathcal{T}(n,s))$
rounds, where $\mathcal{T}(n,s)$ is the running time of the \textsc{RulingSet}()
procedure called with argument $s$.
\label{theorem:Reduction}
\end{theorem}

\section{Computing a $2$-Ruling Set}
\label{sect:2Ruling}

The facility location algorithm in Section \ref{sect:FacLoc} depends on being
able to efficiently compute a $\beta$-ruling set, for small $\beta$, of an
arbitrary spanning subgraph $C'$ of a size-$n$ clique $C$. This section
describes how to compute a $2$-ruling set of $C'$ in a number of rounds which is
$O(\log \log n)$ in expectation.

\subsection{Deterministic Processing of a Sparse Subgraph}

For completeness we start by presenting a simple deterministic subroutine for
efficiently computing a maximal independent set of a \textit{sparse}, induced
subgraph of $C'$. Our algorithm is a simple load-balancing scheme. For a subset
$M \subseteq X$, we use $C'[M]$ to denote the subgraph of $C'$ induced by $M$,
and $E[M]$ and $e[M]$ to denote the set and number, respectively, of edges in
$C'[M]$. The subroutine we present below computes an MIS of $C'[M]$ in
$e[M] / n + O(1)$ rounds. Later, we use this repeatedly in situations where
$e[M] = O(n)$.

We assume that nodes in $X$ have unique identifiers and can therefore be totally
ordered according to these. Let $\rho_i \in \{0, 1, \ldots, n - 1\}$ denote the
rank of node $x_i$ in this ordering. Imagine (temporarily) that edges are
oriented from lower-rank nodes to higher-rank nodes and let $\mathcal{E}(x_i)$
denote the set of outgoing edges incident on $x_i$. Let $d_i$ denote
$|\mathcal{E}(x_i)|$, the outdegree of $x_i$, and let
$D_i = \sum_{j: \rho_j < \rho_i} d_j$ denote the outdegree sum of lower-ranked
nodes.

The subroutine shares the entire topology of $C'[M]$ with all nodes in the
network. To do this efficiently, we map each edge $e \in E[M]$ to a node in $X$.
Information about $e$ will be sent to the node to which $e$ is mapped. Each node
will then broadcast information about all edges that have been mapped to it. See
Algorithm \ref{alg:CCMISAlg}.

\begin{algo}
\textbf{Input:} A subset of nodes $M \subseteq X$\\
\textbf{Output:} An MIS $L$ of $C'[M]$
\begin{tabbing}
......\=a....\=b....\=c....\=d....\=e....\=f....\=g....\=h......\kill
1.\>Each node $x_i$ broadcasts its ID.\\
2.\>$x_i$ computes and broadcasts $d_i$.\\
3.\>$x_i$ assigns a distinct label $\ell(e)$ from
$\{D_i, D_i+1, \ldots, D_i + d_i - 1\}$ to each\\
\>\>incident outgoing edge $e$.\\
4.\>$x_i$ sends each outgoing edge $e$ to the node $x_j$ of rank
$\rho_j = (\ell(e) \mod{n})$.\\
5.\>$x_i$ receives and broadcasts all edges sent to it in the previous step, one
per round.\\
6.\>Each node $x_i$ computes $C'[M]$ from received edges and uses a
deterministic algorithm to\\
\>\>locally compute an MIS $L$.
\end{tabbing}
\caption{Deterministic MIS for Sparse Graphs}
\label{alg:CCMISAlg}
\end{algo}

%Lemma 6
\begin{lemma}
Algorithm \ref{alg:CCMISAlg} computes an MIS $L$ of $C'[M]$ in
$\frac{e[M]}{n} + O(1)$ rounds.
\label{lemma:mis_correct}
\end{lemma}
\begin{proof}
Each node $p_i$ reserves a distinct range
$\{D_i, D_i+1, \ldots, D_i + d_i - 1\}$ of size $d_i$ for labeling the $d_i$
outgoing edges incident on it (Steps 1-3). This implies that the edges in $E[M]$
get unique labels in the range $\{0, 1, \ldots, e[M] - 1\}$. Sending each edge
$e$ to a node $p_j$ with rank $\ell(e)\mod{n}$ means that each node receives at
most $e[M] / n + 1$ edges (Step 4). Note that Steps 1-4 take at most one round
each. Step 5 takes no more than $e[M] / n + 1$ rounds, as this is the maximum
number of edges that can be received by a node in Step 4.
\end{proof}

\subsection{Algorithm}

We are now ready to present an algorithm for computing a $2$-ruling set of $C'$
which is ``super-fast'' in expectation. We show that this algorithm has an
expected running time of $O(\log \log n)$ rounds. The algorithm proceeds in
\textit{Iterations} and in each Iteration some number of nodes leave $C'$. We
measure progress by the number of edges remaining in $C'$, as nodes leave $C'$.

In an Iteration $i$, each node remaining in $C'$ joins a ``Test'' set $T$
independently with probability $q = \sqrt{\frac{n}{m}}$ (Line 6), where
$m = e[C']$ is the number of edges remaining in $C'$ (we also use the notation
$m(i)$ to refer specifically to the value of $m$ at the beginning of, and
during, the $i$th iteration). The probability $q$ is set such that the expected
number of edges in $C'[T]$ is equal to $n$. Once the set $T$ is picked and each
node has broadcast its membership status with respect to $T$, each node can
broadcast its degree in $C'[T]$ and thus allow all nodes to compute $e[C'[T]]$
in a constant number of rounds.

If $e[C'[T]] \leq 4 n$, we use Algorithm \ref{alg:CCMISAlg} to process $C'[T]$
in $O(1)$ rounds, and then we delete $T$ and its neighborhood $N(T)$ from $C'$.
(Lines 7-10). Because $m = e[C']$ decreases, we can raise $q$ (Line 12) while
still having the expected number of edges in $C'[T]$ during the next iteration
bounded above by $n$. See Algorithm \ref{alg:2RulingSetAlg}.

If $e[C'[T]] > 4 n$, then Algorithm \ref{alg:CCMISAlg} is \textit{not} run, no
progress is made, and we proceed to the next iteration. We would like to mention
that the use of this cutoff is for ease of analysis only and is not
fundamentally important to the algorithm.

\begin{algo}
\textbf{Input:} A spanning subgraph $C'$ of the clique $C$\\
\textbf{Output:} A $2$-ruling set $R$ of $C'$
\begin{tabbing}
......\=a....\=b....\=c....\=d....\=e....\=f....\=g....\=h......\kill
1.\>$R := \emptyset$\\
2.\>$m := e[C']$ (Each node $x$ broadcasts its degree in $C'$ to all others,\\
\>\>\>after which each can compute $m$ locally)\\
3.\>$q := \sqrt{\frac{n}{m}}$\\
4.\>\textbf{while} $m > 2 n$ \textbf{do}\\[1mm]
\>\>\textbf{\emph{Start of Iteration:}}\\[1mm]
5.\>\>$T := \emptyset$\\
6.\>\>Each $x \in C'$ joins $T$ independently with probability $q$ and
broadcasts its choice.\\
7.\>\>\textbf{if} $e[C'[T]] \leq 4 n$ \textbf{then}\\
8.\>\>\>All nodes compute an MIS $L$ of $C'[T]$ using Algorithm
\ref{alg:CCMISAlg}.\\
9.\>\>\>$R := R \cup L$\\
10.\>\>\>Remove $(T \cup N(T))$ from $C'$.\\
11.\>\>$m := e[C']$\\
12.\>\>$q := \sqrt{\frac{n}{m}}$\\[1mm]
\>\>\textbf{\emph{End of Iteration}}\\[1mm]
13.\>All nodes compute an MIS $L$ of $C'$ ($C'$ has at most $2 n$ edges
remaining) using AlgorithmgSet \ref{alg:CCMISAlg}.\\
14.\>$R := R \cup L$\\
15.\>Output $R$.
\end{tabbing}
\caption{Super-Fast $2$-Ruling Set}
\label{alg:2RulingSetAlg}
\end{algo}

\subsection{Analysis}

%Lemma 7
\begin{lemma}
Algorithm \ref{alg:2RulingSetAlg} computes a $2$-ruling set of $C'$.
\label{lemma:2rule_correct}
\end{lemma}
\begin{proof}
During any iteration in which $e[C'[T]] \leq 4 n$, the only nodes removed from
$C'$ are those in $T \cup N(T)$. Since we compute an MIS $L$ of $C'[T]$ and
include only these nodes in the final output $R$, currently every node in $T$ is
at distance at most one from a node in $L$ and every node in $N(T)$ is at
distance at most $2$ from a member of $L$. Furthermore, after deletion of
$T \cup N(T)$, no node remaining in $C'$ is a neighbor of any node in $T$, and
therefore no node that can be added to $R$ in the future will have any
adjacencies with nodes added to $R$ in this iteration. So $R$ will remain an
independent set. When the algorithm terminates, all nodes were either in $T$ or
in $N(T)$ at some point, and therefore $R$ is a $2$-ruling set of $C'$.
\end{proof}

\noindent
Define $L_k = n^{1 + 1 / 2^k}$ for $k = 0, 1, \ldots$. Think of the $L_k$'s as
specifying thresholds for the number of edges still remaining in $C'$.
(Initially, i.e., for $k = 0$, $L_0 = n^2$ is a trivial upper bound on the
number of edges in $C'$.) As the algorithm proceeds, we would like to measure
the number of rounds for the number of edges in $C'$ to fall from the threshold
$L_{k-1}$ to the threshold $L_k$. Note that the largest $k$ in which we are
interested is $k = \log \log n$, because for this value of $k$, $L_k = 2 n$ and
if the number of edges falls below this threshold we know how to process what
remains in $C'$ in $O(1)$ rounds. Let $S_k$ denote the smallest iteration index
$i$ at the start of which $e[C'] \leq L_k$. Define $\mathcal{T}(k)$ by
$\mathcal{T}(k) = S_k - S_{k-1}$, i.e., $\mathcal{T}(k)$ is the number of
iterations required to progress from having $L_{k-1}$ edges remaining in $C'$ to
having only $L_k$ edges remaining. We are interested in bounding
$\mathbf{E}(\mathcal{T}(k))$.

%Lemma 8
\begin{lemma}
For each $i \geq 1$, the probability that $e[C'[T]] \leq 4 n$ during the $i$th
iteration is at least $\frac{3}{4}$.
\label{lemma:constant_success}
\end{lemma}
\begin{proof}
In the $i$th iteration, each node remaining in $C'$ joins $T$ independently with
probability $\sqrt{\frac{n}{m}}$, where, as before, $m = m(i) = e[C']$ is the
number of edges remaining in $C'$. Therefore, for any edge remaining in $C'$,
the probability that both of its endpoints join $T$ (and hence that this edge is
included in $C'[T]$) is $\sqrt{\frac{n}{m}} \cdot \sqrt{\frac{n}{m}} =
\frac{n}{m}$. Thus the expected number of such edges, $\mathbf{E}(e[C'[T]])$, is
equal to $e[C'] \cdot \frac{n}{m} = n$. By Markov's inequality,
$\mathbf{P}(e[C'[T]] > 4 n) \leq \frac{n}{4 n} = \frac{1}{4}$.
\end{proof}

%Lemma 9
\begin{lemma}
For each $k \geq 1$, $\mathbf{E}(\mathcal{T}(k)) = O(1)$.
\label{lemma:progress_iterations}
\end{lemma}
\begin{proof}
Suppose that, after $i - 1$ iterations, $m = m(i - 1) = e[C'] \leq L_{k-1}$. We
analyze the expected number of edges remaining in $C'$ after the next iteration.

Let Algorithm \ref{alg:2RulingSetAlg}$_i^*$ refer to the variation on Algorithm
\ref{alg:2RulingSetAlg} in which, during iteration $i$ only, the cutoff value of
$4 n$ in Line 7 is ignored; i.e., an MIS is computed, and nodes subsequently
removed from $C'$, regardless of the number of edges in $C'[T]$ (during
iteration $i$).
We view Algorithms \ref{alg:2RulingSetAlg} and \ref{alg:2RulingSetAlg}$_i^*$
as being \textit{coupled} through the first $i$ iterations; in other words,
%Furthermore, since Algorithm \ref{alg:2RulingSetAlg} and Algorithm
%\ref{alg:2RulingSetAlg}$_i^*$ have the same behavior through the first $i - 1$
%iterations, let us work in a probability space $\mathcal{P}$ in which these two
%algorithms are \textit{coupled} through the first $i$ iterations - in other
%words, in the framework $\mathcal{P}$, they use the same random variables in
%selecting which nodes join $T$ during the first $i$ iterations. Consequently,
%within $\mathcal{P}$, Algorithms \ref{alg:2RulingSetAlg} and
%\ref{alg:2RulingSetAlg}$_i^*$ 
the two algorithms have the same history and make the same progress
during the first $i - 1$ iterations.

Let $m^*(j)$ be the random variable which is the number of edges remaining at
the beginning of iteration $j$ with Algorithm \ref{alg:2RulingSetAlg}$_i^*$.
Let $\mathrm{deg}_{j}^*(x)$ be the degree of $x$ in $C'$ under Algorithm
\ref{alg:2RulingSetAlg}$_i^*$ at the beginning of iteration $j$. We can bound
the expected value of $m^*(i + 1)$ by bounding, for each $x$,
$\mathbf{E}(\mathrm{deg}_{i+1}^*(x))$.

In turn, $\mathbf{E}(\mathrm{deg}_{i+1}^*(x))$ can be bounded above
by the degree of $x$ at the beginning of the $i$th iteration,
$\mathrm{deg}_i^*(x)$, multiplied by the probability that $x$ remains in $C'$
after iteration $i$. (The degree of $x$ can be considered to be $0$ if $x$ has
been removed from $C'$, for the purpose of computing the number of edges
remaining in the subgraph. Furthermore, under Algorithm
\ref{alg:2RulingSetAlg}$_i^*$, we may upper bound the probability of $x$
remaining in $C'$ after the $i$th iteration by the probability that no neighbor
of $x$ joins $T$ during iteration $i$.

Under Algorithm \ref{alg:2RulingSetAlg}$_i^*$, then, the expected number of
edges remaining in $C'$ after iteration $i$ is
\begin{align*}
\mathbf{E}\left(m^*(i + 1)\right) &= \mathbf{E}\left(\frac{1}{2} %
\sum\limits_{x \in C} \mathrm{deg}_{i+1}^*(x)\right)\\[1mm]
&= \frac{1}{2} \sum\limits_{x \in C} %
\mathbf{E}\left(\mathrm{deg}_{i+1}^*(x)\right)\\[1mm]
&\leq \frac{1}{2} \sum\limits_{x \in C} %
\mathbf{P}\left(x \notin T \cup N(T) %
\text{ under Alg. \ref{alg:2RulingSetAlg}$_i^*$}\right) \cdot %
\mathrm{deg}_i^*(x)\\[1mm]
&\leq \frac{1}{2} \sum\limits_{x \in C} %
\left(1 - \sqrt{\frac{n}{m^*(i)}}\right)^{\mathrm{deg}_i^*(x)} \cdot %
\mathrm{deg}_i^*(x)\\[1mm]
&\leq \frac{1}{2} \sum\limits_{x \in C} %
\left(e^{-\sqrt{\frac{n}{m^*(i)}}}\right)^{\mathrm{deg}_i^*(x)} \cdot %
\mathrm{deg}_i^*(x)\\[1mm]
&\leq \frac{1}{2} \sqrt{\frac{m^*(i)}{n}} \cdot \sum\limits_{x \in C} %
\left[\sqrt{\frac{n}{m^*(i)}} \mathrm{deg}_i^*(x) \cdot %
e^{-\sqrt{\frac{n}{m^*(i)}} \mathrm{deg}_i^*(x)}\right]\\
\end{align*}
Note that $z \cdot e^{-z} \leq \frac{1}{e}$ for all $z \in \mathbb{R}$, so the
summand in this last quantity can be replaced by $\frac{1}{e}$. We then have
\begin{align*}
\mathbf{E}\left(\frac{1}{2} \sum\limits_{x \in C} %
\mathrm{deg}_{i+1}^*(x)\right) &\leq \frac{1}{2} \sqrt{\frac{m^*(i)}{n}} \cdot %
\sum\limits_{x \in C} \frac{1}{e}\\[1mm]
&= \frac{1}{2} \sqrt{\frac{m^*(i)}{n}} \cdot \frac{n}{e}\\[1mm]
&= \frac{1}{2 e} \sqrt{n \cdot m^*(i)}\\
\end{align*}
Since Algorithms \ref{alg:2RulingSetAlg} and
\ref{alg:2RulingSetAlg}$_i^*$ are coupled through the first $i - 1$
iterations, $m^*(i) = m(i)$ and this last quantity satisfies
$$\frac{1}{2 e} \sqrt{n \cdot m^*(i)} = \frac{1}{2 e} \sqrt{n \cdot m(i)} \leq \frac{1}{2 e} \sqrt{n^{2 + 1 / 2^{k-1}}} = \frac{L_k}{2 e}.$$
Therefore, the expected value of $m^*(i + 1)$ is bounded above by
$\frac{1}{2 e} L_k$, and so by Markov's inequality,
$\mathbf{P}(m^*(i + 1) > L_k \mid m^*(i) \leq L_{k-1}) \leq
\frac{L_k}{2 e \cdot L_k} < \frac{1}{4}$.

\noindent
As mentioned before, Algorithm \ref{alg:2RulingSetAlg} and Algorithm
\ref{alg:2RulingSetAlg}$_i^*$ have the same history through the first $i - 1$ iterations. As
well, they also have the same history through the $i$th iteration in the event
that $e[C'[T]] \leq 4 n$ during the $i$th iteration. During iteration $i$, if
$e[C'[T]] > 4 n$, then Algorithm \ref{alg:2RulingSetAlg}$_i^*$ may still make
progress (adding nodes to the $2$-ruling set), whereas Algorithm
\ref{alg:2RulingSetAlg} makes none.

Let $E_1 = \{e[C'[T]] > 4 n \text{ in iteration } i\}$.
By Lemma \ref{lemma:constant_success},
$\mathbf{P}(E_1 \mid m(i) \leq L_{k-1}) \leq \frac{1}{4}$.
Let $E_2 = \{m^*(i + 1) \leq L_k\}$. By the earlier
analysis, $\mathbf{P}(E_2 \mid m^*(i) \leq L_{k-1}) > \frac{3}{4}$. 
Thus the event $E_2 \setminus E_1$ conditioned on $m^*(i) \leq L_{k-1}$ is such that (i) Algorithm
\ref{alg:2RulingSetAlg} is identical to (with the same history and behavior as)
Algorithm \ref{alg:2RulingSetAlg}$_i^*$ through iteration $i$, and (ii)
$m(i + 1) = m^*(i + 1) \leq L_k$. Thus,
\begin{align*}
\mathbf{P}(E_2 \setminus E_1 \mid m^*(i) \leq L_{k-1}) &\geq %
\mathbf{P}(E_2 \mid m^*(i) \leq L_{k-1}) - %
\mathbf{P}(E_1 \mid m^*(i) \leq L_{k-1})\\[1mm]
&\geq \frac{3}{4} - \frac{1}{4}\\[1mm]
&= \frac{1}{2}\\
\end{align*}
Thus, given that $m(i) = m^*(i) \leq L_{k-1}$, with probability at least
$\frac{1}{2}$ we have $m(i + 1) = m^*(i + 1) \leq L_k$, and Algorithm
\ref{alg:2RulingSetAlg} makes progress by one level. Since this holds for every
$i$, the expected number of additional iterations required under Algorithm
\ref{alg:2RulingSetAlg} before $m \leq L_k$ is a (small) constant ($2$), and
hence $\mathbf{E}(\mathcal{T}(k)) = O(1)$.
\end{proof}

%Theorem 2
\begin{theorem}
Algorithm \ref{alg:2RulingSetAlg} computes a $2$-ruling set on the subgraph $C'$
of the clique $C$ and has an expected running time of $O(\log \log n)$ rounds.
\label{theorem:2rule_time}
\end{theorem}
\begin{proof}
By Lemma \ref{lemma:2rule_correct}, the output $R$ is a $2$-ruling set of $C'$.
To bound the expected running time, observe that 
$$L_{\log \log n} = n^{1 + 1 / 2^{\log \log n}} = n^{1 + 1 / \log n} = n^{1 + \log_n 2} = 2 n,$$
which is the point at which Algorithm \ref{alg:2RulingSetAlg} exits the
\textbf{while} loop and runs one deterministic iteration to process the
remaining (sparse) graph. Now, given some history, $\mathcal{T}(k)$ is the
random variable which is the number of iterations necessary to progress from
having at most $L_{k-1}$ edges remaining in $C'$ to having at most $L_k$ edges
remaining, so let $I_{k,j}$ be the running time, in rounds, of the $j$th such
iteration (for $j = 1, \ldots, \mathcal{T}(k)$; as well, $\mathcal{T}(k)$ may be
$0$). Note that $I_{k,j}$ is bounded by a constant due to the cutoff condition
of Line 7.

The running time of Algorithm \ref{alg:2RulingSetAlg} is thus
$O(1) + \sum_{k=1}^{\log \log n} \sum_{j=1}^{\mathcal{T}(k)} I_{k,j}$, and the
expected running time can be described as
\begin{align*}
\mathbf{E}\left(O(1) + \sum\limits_{k=1}^{\log \log n} %
\sum\limits_{j=1}^{\mathcal{T}(k)} I_{k,j}\right) &= %
O(1) + \sum\limits_{k=1}^{\log \log n} %
\mathbf{E}\left(\sum\limits_{j=1}^{\mathcal{T}(k)} I_{k,j}\right)\\[1mm]
&\leq O(1) + \sum\limits_{k=1}^{\log \log n} %
\mathbf{E}\left(\sum\limits_{j=1}^{\mathcal{T}(k)} O(1)\right)\\[1mm]
&= O(1) + \sum\limits_{k=1}^{\log \log n} %
O\left(\mathbf{E}\left(\mathcal{T}(k)\right)\right)\\[1mm]
&= O(1) + \sum\limits_{k=1}^{\log \log n} O(1)\\[1mm]
&= O(\log \log n)\\
\end{align*}
which completes the proof.
\end{proof}

Using Algorithm \ref{alg:2RulingSetAlg} as a specific instance of the procedure
\textsc{RulingSet}() for $s = 2$ and combining Theorems \ref{theorem:Reduction}
and \ref{theorem:2rule_time} leads us to the following result.

%Theorem 3
\begin{theorem}
There exists an algorithm that solves the \textsc{CliqueFacLoc} problem with an
expected running time of $O(\log \log n)$ communication rounds.
\label{theorem:MainResult}
\end{theorem}

\section{Concluding Remarks}

It is worth noting that under special circumstances an $O(1)$-ruling set of a
spanning subgraph of a clique can be computed even more quickly. For example, if
the subgraph of $C$ induced by the nodes in class $V_k$ is growth-bounded for
each $k$, then we can use the Schneider-Wattenhofer
\cite{SchneiderWattenhofer08} result to compute an MIS for $H_k$ in
$O(\log^* n)$ rounds (in the $\mathcal{CONGEST}$ model). It is easy to see that
if the metric space $(X, D)$ has constant doubling dimension, then $H_k$ would
be growth-bounded for each $k$. A Euclidean space of constant dimension has
constant doubling dimension and therefore this observation applies to
constant-dimension Euclidean spaces. This discussion is encapsulated in the
following theorem.

%Theorem 4
\begin{theorem}
The \textsc{CliqueFacLoc} problem can be solved in $O(\log^* n)$ rounds on a
metric space of constant doubling dimension.
\end{theorem}

The lack of lower bounds for problems in the $\mathcal{CONGEST}$ model on a
clique network essentially means it might be possible to solve
\textsc{CliqueFacLoc} via even faster algorithms. It may, for example, be
possible to compute an $s$-ruling set, for constant $s > 2$, in time
$o(\log \log n)$; this would lead to an even faster constant-approximation
for \textsc{CliqueFacLoc}. This is a natural avenue of future research suggested
by this work.

Another natural question suggested by our expected-$O(\log \log n)$-round
algorithm for computing a $2$-ruling set on a subgraph of a clique is whether or
not an algorithm this fast exists for computation of a maximal independent set
($1$-ruling set), in the same setting, also. The analysis of our algorithm
depends very significantly on the fact that when a node is added to our
solution, not only its neighbors but \textit{all} nodes in its $2$-neighborhood
are removed. Thus MIS computation, and additionally $(\Delta + 1)$-coloring, in
$O(\log \log n)$ rounds on a spanning subgraph of a clique are examples of other
open problems suggested by this work.

\bibliographystyle{splncs03}
\bibliography{DistCompCliqueFacLoc}

\begin{thebibliography}{10}
\providecommand{\url}[1]{\texttt{#1}}
\providecommand{\urlprefix}{URL }

\bibitem{Balinski66}
Balinski, M.: {On finding integer solutions to linear programs}. In:
  {Proceedings of IBM Scientific Computing Symposium on Combinatorial
  Problems}. pp. 225--248 (1966)

\bibitem{BCIS05}
B\u{a}doiu, M., Czumaj, A., Indyk, P., Sohler, C.: {Facility location in
  sublinear time}. In: {ICALP}. pp. 866--877 (2005)

\bibitem{CNWBook}
Cornuejols, G., Nemhouser, G., Wolsey, L.: {Discrete Location Theory}. Wiley
  (1990)

\bibitem{EHK77}
Eede, M.V., Hansen, P., Kaufman, L.: {A plant and warehouse location problem}.
  Operational Research Quarterly  28(3),  547--554 (1977)

\bibitem{FrankBook}
Frank, C.: {Algorithms for Sensor and Ad Hoc Networks}. Springer (2007)

\bibitem{GLS06}
Gehweiler, J., Lammersen, C., Sohler, C.: {A distributed O(1)-approximation
  algorithm for the uniform facility location problem}. In: {Proceedings of the
  eighteenth annual ACM symposium on Parallelism in algorithms and
  architectures}. pp. 237--243. {SPAA '06}, ACM, ACM Press, New York, NY, USA
  (2006)

\bibitem{GuhaKhuller98}
Guha, S., Khuller, S.: {Greedy strikes back: Improved facility location
  algorithms}. In: {Proceedings of the ninth annual ACM-SIAM symposium on
  Discrete algorithms}. pp. 649--657. Society for Industrial and Applied
  Mathematics (1998)

\bibitem{HamburgerKuehn63}
Hamburger, M.J., Kuehn, A.A.: {A heuristic program for locating warehouses}.
  Management science  9(4),  643--666 (1963)

\bibitem{KMW10}
Kuhn, F., Moscibroda, T., Wattenhofer, R.: {Local Computation: Lower and Upper
  Bounds}. CoRR  abs/1011.5470 (2010)

\bibitem{Lenzen12}
Lenzen, C.: Optimal deterministic routing and sorting on the congested clique.
  CoRR  abs/1207.1852 (2012)

\bibitem{LenzenWattenhofer10}
Lenzen, C., Wattenhofer, R.: Brief announcement: Exponential speed-up of local
  algorithms using non-local communication. In: Proceeding of the 29th ACM
  SIGACT-SIGOPS symposium on Principles of distributed computing. pp. 295--296.
  ACM (2010)

\bibitem{LenzenWattenhofer11}
Lenzen, C., Wattenhofer, R.: Tight bounds for parallel randomized load
  balancing: extended abstract. In: STOC. pp. 11--20 (2011)

\bibitem{Li11}
Li, S.: {A 1.488-approximation algorithm for the uncapacitated facility
  location problem}. In: {Proceedings of the 38th international colloquium on
  automata, languages and programming}. pp. 77--88. {ICALP '11},
  Springer-Verlag, Berlin, Heidelberg (2011)

\bibitem{LPPP05}
Lotker, Z., Patt-Shamir, B., Pavlov, E., Peleg, D.: {Minimum-weight spanning
  tree construction in O(log log n) communication rounds}. SIAM J. Comput.
  35(1),  120--131 (2005)

\bibitem{LPP06}
Lotker, Z., Patt-Shamir, B., Peleg, D.: Distributed {MST} for constant diameter
  graphs. Distributed Computing  18(6),  453--460 (2006)

\bibitem{MettuPlaxton03}
Mettu, R.R., Plaxton, C.G.: {The online median problem}. SIAM J. Comput.
  32(3),  816--832 (2003)

\bibitem{MoscibrodaWattenhofer05}
Moscibroda, T., Wattenhofer, R.: {Facility location: distributed
  approximation}. In: {Proceedings of the twenty-fourth annual ACM symposium on
  Principles of distributed computing}. pp. 108--117. ACM, ACM Press, New York,
  NY, USA (2005)

\bibitem{PanditPemmarajuICDCN09}
Pandit, S., Pemmaraju, S.V.: {Finding facilities fast}. Distributed Computing
  and Networking pp. 11--24 (2009)

\bibitem{PanditPemmaraju09}
Pandit, S., Pemmaraju, S.V.: {Return of the primal-dual: distributed metric
  facility location}. In: {Proceedings of the 28th ACM symposium on Principles
  of distributed computing}. pp. 180--189. {PODC '09}, ACM, ACM Press, New
  York, NY, USA (2009)

\bibitem{PanditPemmaraju10}
Pandit, S., Pemmaraju, S.V.: {Rapid randomized pruning for fast greedy
  distributed algorithms}. In: {Proceedings of the 29th ACM SIGACT-SIGOPS
  symposium on Principles of distributed computing}. pp. 325--334. ACM (2010)

\bibitem{PattShamirTeplitsky11}
Patt-Shamir, B., Teplitsky, M.: The round complexity of distributed sorting:
  extended abstract. In: PODC. pp. 249--256. ACM Press (2011)

\bibitem{PelegBook}
Peleg, D.: {Distributed computing: a locality-sensitive approach}, vol.~5.
  Society for Industrial and Applied Mathematics (2000)

\bibitem{SchneiderWattenhofer08}
Schneider, J., Wattenhofer, R.: {A log-star distributed maximal independent set
  algorithm for growth-bounded graphs}. In: {Proceedings of the twenty-seventh
  ACM symposium on Principles of distributed computing}. pp. 35--44. ACM (2008)

\bibitem{Stollsteimer63}
Stollsteimer, J.F.: {A working model for plant numbers and locations}. Journal
  of Farm Economics  45(3),  631--645 (1963)

\end{thebibliography}

\end{document}